\documentclass[10pt,twocolumn,letterpaper]{article}

\usepackage{cvpr}
\usepackage{times}
\usepackage{epsfig}
\usepackage{graphicx}
\usepackage{amsmath}
\usepackage{amssymb}
\usepackage{makecell}

\usepackage{subfiles}
\usepackage[normalem]{ulem}
\usepackage{microtype}
\usepackage{url}
\usepackage{subfigure}
\usepackage{booktabs}
\usepackage{amsfonts,amsthm,bbm,mathrsfs}
\usepackage{algorithm,algorithmic,bm,color}
\usepackage{multirow}
\usepackage{wrapfig}
\usepackage{adjustbox}
\usepackage{graphics, color,  dsfont}
\usepackage[nottoc]{tocbibind}
\usepackage[square,numbers]{natbib} 
\usepackage[flushleft]{threeparttable} 
\usepackage{boldline} 
\usepackage[table]{xcolor} 
\def\comment#1{}

\newcommand{\argmin}{\arg\!\min}

\newcommand{\Bmat}{{\boldsymbol B}}
\newcommand{\Cmat}{{\boldsymbol C}}
\newcommand{\Dmat}{{\boldsymbol D}}

\newcommand{\Hmat}[0]{{{\boldsymbol H}}}
\newcommand{\Imat}{{\boldsymbol I}}

\newcommand{\Qmat}[0]{{{\boldsymbol Q}}}
\newcommand{\Rmat}[0]{{{\boldsymbol R}}}

\newcommand{\Xmat}{{\boldsymbol X}}
\newcommand{\Ymat}[0]{{{\boldsymbol Y}}}
\newcommand{\Zmat}{{\boldsymbol Z}}

\newcommand{\dv}{\boldsymbol{d}}

\newcommand{\uv}[0]{{\boldsymbol{u}}}
\newcommand{\vv}{\boldsymbol{v}}

\newcommand{\xv}{\boldsymbol{x}}
\newcommand{\yv}{\boldsymbol{y}}

\newcommand{\zv}{\boldsymbol{z}}

\newcommand{\Thetamat}{\boldsymbol{\Theta}}

\newcommand{\Phimat}{\boldsymbol{\Phi}}

\newcommand{\thetav}{\boldsymbol{\theta}}

\newcommand{\ts}{^{\top}}
\newcommand{\inv}{^{-1}}

\newtheorem{definition}{Definition}
\newtheorem{theorem}{Theorem}
\newtheorem{corollary}{Corollary}

\newtheorem{lemma}{Lemma}

\newtheorem{assumption}{Assumption}

\usepackage[breaklinks=true,bookmarks=false]{hyperref}

\cvprfinalcopy 

\def\cvprPaperID{6137} 

\ifcvprfinal\fi
\begin{document}

\title{Plug-and-Play Algorithms for Large-scale Snapshot Compressive Imaging}

\author{Xin Yuan\\
Bell Labs\\
NJ USA\\
{\tt\small xyuan@bell-labs.com}
\and
Yang Liu\\
MIT\\
MA USA\\
{\tt\small yliu12@mit.edu}
\and
Jinli Suo \qquad\qquad ~Qionghai Dai\\
Dept. of Automation \& Institute for Brain and \\Cognitive Sciences, Tsinghua Univ., Beijing China\\
{\tt\small \{jlsuo,daiqh\}@tsinghua.edu.cn}
}

\maketitle

\begin{abstract}
Snapshot compressive imaging (SCI) aims to capture the high-dimensional (usually 3D) images using a 2D sensor (detector) in a single snapshot.
 Though enjoying the advantages of low-bandwidth, low-power and low-cost, applying SCI to large-scale problems (HD or UHD videos) in our daily life is still challenging.
 The bottleneck lies in the reconstruction algorithms; they are either too slow (iterative optimization algorithms) or not flexible to the encoding process (deep learning based end-to-end networks).  
 In this paper, we develop fast and flexible algorithms for SCI based on the plug-and-play (PnP) framework.
 In addition to the widely used PnP-ADMM method, 
 we further propose the PnP-GAP (generalized alternating projection) algorithm with a lower computational workload and prove the convergence\footnote{We have found an error in the proof of the camera ready version of the CVPR paper in the CVF website.
 		Specifically, the lower bound in Eq. (25) is wrong. Following this, the proved global converge of PnP-GAP does not hold. In this new version, we show another convergence proof of PnP-GAP.} of PnP-GAP under the SCI hardware constraints. 
 By employing deep denoising priors, we first time show that PnP can recover 
 a UHD color video ($3840\times 1644\times 48$ with PNSR above 30dB) from a snapshot 2D measurement. 
 Extensive results on both simulation and real datasets verify the superiority of our proposed algorithm. The code is available at \url{https://github.com/liuyang12/PnP-SCI}
\end{abstract}

\section{Introduction \label{Sec:Intro}}

Computational imaging~\cite{Altmanneaat2298,Mait18CI} constructively combines optics, electronics and algorithms for optimized performance~\cite{BradyNature12,Brady18Optica,Ouyang2018DeepLM} or to provide new abilities~\cite{Brady15AOP,Llull2015_book,Qiao:19_DH3D,Tsai15OL} to imaging systems. 
One important branch of computational imaging with promising applications is snapshot compressive imaging (SCI)~\cite{Patrick13OE,Wagadarikar08CASSI}, which utilized a two-dimensional (2D) camera to capture 3D  video or spectral data. 
Different from conventional cameras, such imaging systems adopt {sampling} on \emph{a set of consecutive images}--video frames (\eg, CACTI~\cite{Patrick13OE,Yuan14CVPR}) or spectral channels (\eg, CASSI~\cite{Wagadarikar09CASSI})--in accordance with the sensing matrix and {integrating} these sampled signals along time  or spectrum to obtain the final compressed measurements. With this technique, SCI systems~\cite{Gehm07,Hitomi11ICCV,Reddy11CVPR,Tsai15COSI,Wagadarikar08CASSI,Wagadarikar09CASSI,Yuan14CVPR} can capture the high-speed motion~\cite{Sun16OE,Sun17OE,Yuan&Pang16_ICIP,Yuan16BOE,Yuan17AO,Yuan13ICIP} and high-resolution spectral information~\cite{Miao19ICCV,Yuan15JSTSP,Renna16_TIT_side} but with low memory, low bandwidth, low power and potentially low cost. 
In this work, we focus on video SCI reconstruction. 

\begin{figure*}[t]
 	\begin{center}
 		\includegraphics[width=\linewidth]{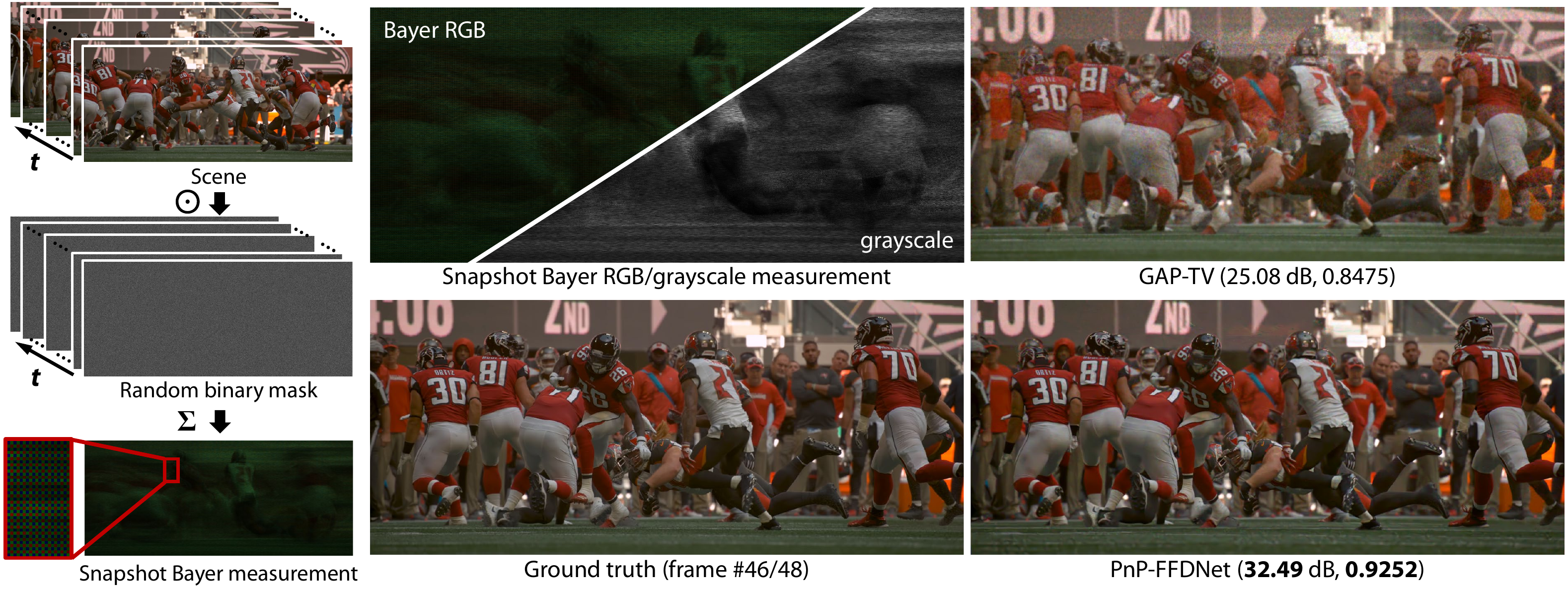}  
 	\end{center}
 	\vspace{-4mm}
 	\caption{Sensing process of video SCI (left) and the reconstruction results using the proposed PnP-FFDNet (bottom-right). The captured image (middle-top) size is {\bf UHD ($\mathbf{3840\times1644}$)} and {\bf 48 frames} are recovered from {\bf a snapshot measurement}. GAP-TV (top-right) takes 180 mins and PnP-FFDNet takes 55 mins for the reconstruction. All other methods are too slow (more than 12 hours) to be used. 
 	}
 	\label{fig:performance_4k}
 \end{figure*}

In parallel to the hardware development, various algorithms have been employed and developed for SCI reconstruction.
In addition to the widely used TwIST~\cite{Bioucas-Dias2007TwIST}, Gaussian Mixture Model (GMM) in~\cite{Yang14GMMonline,Yang14GMM} based algorithms model the pixels within a spatial-temporal patch by a GMM. GAP-TV~\cite{Yuan16ICIP_GAP} adopts the idea of total variance minimization under the generalized alternating projection (GAP)~\cite{Liao14GAP} framework. 
Most recently, DeSCI proposed in~\cite{Liu18TPAMI}  has led to state-of-the-art results.
However, the slow speed of DeSCI precludes its real applications, especially to the HD ($1280\times720$), FHD ($1920\times1080$) or UHD ($3840\times1644$  in Fig.~\ref{fig:performance_4k} and $3840\times2160$ in Fig.~\ref{fig:comp_largescale}) videos, which are getting popular in our daily life.
Recall that DeSCI needs more than one hour to reconstruct a $256\times256\times8$ video from a snapshot measurement.
GAP-TV, by contrast, as a fast algorithm, cannot provide good reconstruction to be used in real applications (in general, this needs the PSNR$\ge30$dB).
An alternative solution is to train an end-to-end network~\cite{Ma19ICCV,Qiao2020_APLP} to reconstruct the videos for the SCI system. On one hand, this approach can finish the task within seconds and by appropriate design of multiple GPUs, an end-to-end sampling and reconstruction framework can be built. On the other hand, this method loses the {\em robustness} of the network since whenever the sensing matrix (encoding process) changes, a new network has to be re-trained. Moreover, it cannot be readily used in adaptive sensing~\cite{Yuan13ICIP}.

Therefore, it is desirable to devise an {\em efficient} and {\em flexible} algorithm for SCI reconstruction, especially for large-scale problems. This will pave the way of applying SCI in our daily life~\cite{Yuan2020SPM}.
In order to solve the {\em trilemma of speed, accuracy and flexibility} for SCI reconstruction, this paper makes the following contributions:
\begin{enumerate}
	\item Inspired by the plug-and-play (PnP) alternating direction method of multiplier (ADMM)~\cite{Chan2017PlugandPlayAF} framework, we extend PnP-ADMM to SCI and show that PnP-ADMM converges to a {fixed point} by considering the hardware constraints and the special structure of the sensing matrix~\cite{Jalali19TIT_SCI} in SCI.
	\item We propose an efficient PnP-GAP algorithm by using various {\em bounded denoisers} (Fig.~\ref{fig:demo}) into the GAP~\cite{Liao14GAP}, which has a lower computational workload than PnP-ADMM. 
	We further prove that, under proper assumptions, the solution of PnP-GAP will converge. 
	\item 
	By integrating the deep image denoiser, \eg, the {\em fast} and {\em flexible} FFDNet~\cite{Zhang18TIP_FFDNet} into PnP-GAP,	we show that a FHD video ($1920\times1080\times24$) can be recovered from a snapshot measurement (Fig.~\ref{fig:comp_largescale}) within 2 minutes with PSNR close to 30dB using a single GPU plus a normal computer.
	Compared with an end-to-end network~\cite{Ma19ICCV}, dramatic resources have been saved (no re-training is required). 
	This further makes the UHD compression using SCI to be feasible (a $3840\times1644\times 48$ video is reconstructed with PSNR above 30dB in Fig.~\ref{fig:performance_4k}). To our best knowledge, this is the first time that SCI is used in these large-scale problems.
	\item We apply our developed PnP algorithms to extensive simulation and real datasets (captured by real SCI cameras) to verify the efficiency and robustness of our proposed algorithms. We show that the proposed algorithm can obtain results on-par with DeSCI but with a significant reduction of computational time.  
\end{enumerate}

\begin{figure}[!htbp]
	\begin{center}
		\includegraphics[width=0.95\linewidth]{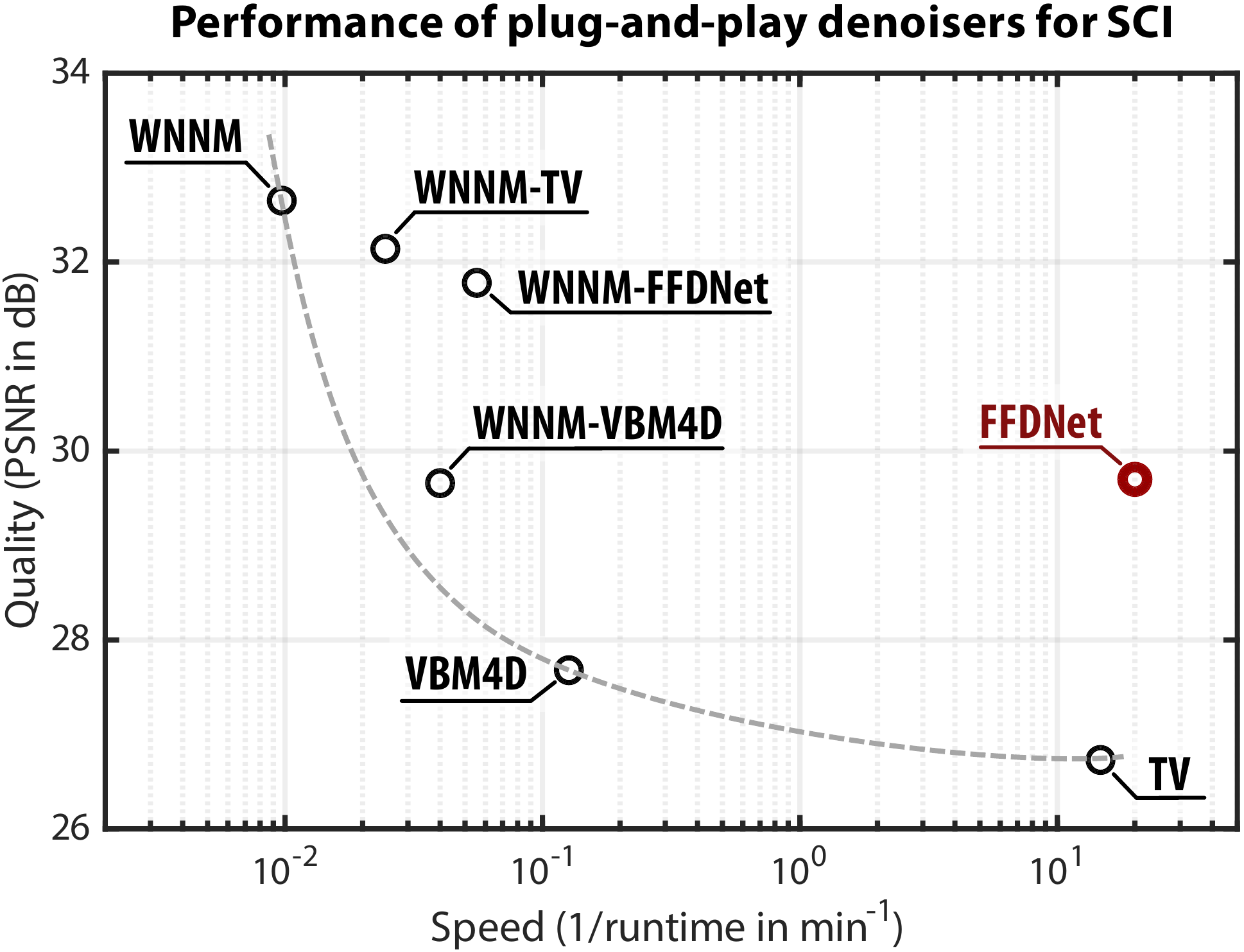}
	\end{center}
	\caption{Trade-off of quality and speed of various plug-and-play denoising algorithms for SCI reconstruction.}
	\label{fig:demo}
\end{figure}

The rest of this paper is organized as follows. Sec.~\ref{Sec:SCImodel} reviews the mathematical model of video SCI. Sec.~\ref{Sec:PnP_ADMM} develops the PnP-ADMM under the SCI hardware constraints and shows that PnP-ADMM converges to a fixed point. Sec.~\ref{Sec:PnP_GAP} develops the PnP-GAP algorithm and proves its global convergence. 
Sec.~\ref{Sec:P3} integrates various denoisers into to the PnP framework for SCI reconstruction. 
Extensive results of both (benchmark and large-scale) simulation and real data are presented in Sec.~\ref{Sec:results} and Sec.~\ref{Sec:Con} concludes the paper.

\paragraph{Related Work \label{Sec:Related}}
SCI systems have been developed to capture  3D spectral images~\cite{Cao16SPM,Renna16_TIT_side,Yuan15JSTSP}, videos~\cite{Hitomi11ICCV,Qiao2020_CACTI,Patrick13OE,Llull2015_book,Reddy11CVPR,Sun16OE,Sun17OE,Tsai15OL,Yuan14CVPR,Yuan17_COSI_rgbD}, high dynamic range~\cite{Yuan_16_OE}, depth~\cite{Llull15Optica,Yuan16AO,Yuan16COSI} and polarization~\cite{Tsai15OE} images, \etc. 
From the algorithm side, in addition to sparsity~\cite{Yuan2020TMM,Yuan14TSP,Yuan16SJ,Yuan14Tree,Yuan17_CSCFA_ICIP,Yuan18SP,Zhang18CVPR,Zha2020_TIP_RRC,Zha2018NonconvexW} based algorithms, GMM~\cite{Yang14GMMonline,Yang14GMM,Yuan15GMM} and  GAP-TV~\cite{Yuan16ICIP_GAP} have been proposed. 
As mentioned above, DeSCI~\cite{Liu18TPAMI} has led to state-of-the-art results.
%
Inspired by deep learning on image restoration~\cite{zhang2017beyond}, researchers have started using deep learning in computational imaging~\cite{Iliadis18DSPvideoCS,Jin17TIP,Kulkarni2016CVPR,LearningInvert2017,George17lensless,Miao:19_DH3D,Yuan18OE}. Some networks have been proposed for SCI reconstruction~\cite{Ma19ICCV,Miao19ICCV,Qiao2020_APLP,Yoshida18ECCV}. 
Different from these methods, in this work, we integrate various denoisers into PnP framework~\cite{Chan2017PlugandPlayAF,Ryu2019PlugandPlayMP} for SCI reconstruction, thus to provide {\em efficient} and {\em flexible} algorithms for SCI. Our PnP algorithms can not only provide excellent results but also are robust to different coding process and thus can be used in adaptive sensing.

\section{Mathematical Model of SCI~\label{Sec:SCImodel}}
As depicted in Fig.~\ref{fig:performance_4k}, in the video SCI system (\eg, CACTI)~\cite{Patrick13OE}, consider that a  $B$-frame video $\Xmat \in \mathbb{R}^{n_x \times n_y \times B}$ is modulated  and compressed by $B$ sensing matrices (masks) $\Cmat\in \mathbb{R}^{n_x \times n_y \times B}$, and the measurement frame $\Ymat \in \mathbb{R}^{n_x\times n_y} $ can be expressed as~\cite{Patrick13OE,Yuan14CVPR}
\begin{equation}\label{Eq:System}
 \Ymat = \sum_{b=1}^B \Cmat_b\odot \Xmat_b + \Zmat,
\end{equation}
where $\Zmat \in \mathbb{R}^{n_x \times n_y }$ denotes the noise; $\Cmat_b = \Cmat(:,:,b)$ and $\Xmat_b = \Xmat(:,:,b) \in \mathbb{R}^{n_x \times n_y}$ represent the $b$-th sensing matrix (mask) and the corresponding video frame, respectively; $\odot$ denotes the Hadamard (element-wise) product. 
Mathematically, the measurement in \eqref{Eq:System} can be expressed by 
\begin{equation}\label{Eq:ghf}
\yv = \Hmat \xv + \zv,
\end{equation}
where $\yv = \text{Vec}(\Ymat) \in \mathbb{R}^{n_x n_y}$ and $\zv= \text{Vec}(\Zmat) \in \mathbb{R}^{n_x n_y}$. Correspondingly, the video signal $\xv \in \mathbb{R}^{n_x n_y B}$ is
\begin{equation}
\xv = \text{Vec}(\Xmat) = [\text{Vec}(\Xmat_1)\ts,..., \text{Vec}(\Xmat_B)\ts]\ts.
\end{equation}
Unlike traditional compressive sensing~\cite{Candes05compressed,donoho2006compressed}, the sensing matrix $\Hmat \in \mathbb{R}^{n_x n_y \times n_x n_y B}$ in video SCI is sparse and is a concatenation of diagonal matrices
\begin{equation}\label{Eq:Hmat_strucutre}
\Hmat = [\Dmat_1,...,\Dmat_B].
\end{equation}
where $\Dmat_b = \text{diag}(\text{Vec}(\Cmat_b)) \in {\mathbb R}^{n \times n}$ with $n = n_x n_y$ , for $b =1,\dots B$.
Consequently, the {\em sampling rate} here is equal to  $1/B$. It has been proved recently in~\cite{Jalali18ISIT,Jalali19TIT_SCI} that the reconstruction error of SCI is bounded even when $B>1$.

In the color video case, as shown in Figs.~\ref{fig:performance_4k}, \ref{fig:comp_largescale} and \ref{fig:real_color_hammer}, the raw data captured by the generally used Bayer pattern sensors have ``RGGB" channels. 
Since the mask is imposed on each pixel, the generated measurement can be treated as a grayscale image as in Fig.~\ref{fig:real_chopperwheel} and when it is shown in color, the demosaicing procedure cannot generate the right color due to mask modulation (Fig.~\ref{fig:comp_largescale}). Therefore, during reconstruction, we first recover each of these four channels independently and then perform demosaicing in the reconstructed videos. The final demosaiced RGB video is the desired signal~\cite{Yuan14CVPR}.

\section{Plug-and-Play ADMM for SCI~\label{Sec:PnP_ADMM}}
The inversion problem of SCI can be modeled as
\begin{equation}
 {\hat \xv} = \argmin_{\xv} f(\xv) + \lambda g(\xv), \label{Eq:uncontr}
\end{equation}
where $f(\xv)$ can be seen as the loss of the forward imaging model, \ie, $\|\yv-\Hmat\xv\|_2^2$ and $g(\xv)$ is a prior being used. This prior is usually playing the role of a regularizer.

\subsection{Review the Plug-and-Play ADMM in~\cite{Chan2017PlugandPlayAF}}
Via using the ADMM framework~\cite{Boyd11ADMM}, by introducing an auxiliary parameter $\vv$, the unconstrained optimization in Eq.~\eqref{Eq:uncontr} can be converted into
\begin{equation} \label{Eq:ADMM_xv}
({\hat \xv}, {\hat \vv}) = \argmin_{\xv,\vv} f(\xv) + \lambda g(\vv), {\text{ subject to }} \xv = \vv.
\end{equation}  
This minimization can be solved by the following sequence of sub-problems
\begin{align}
\xv^{(k+1)} &= \argmin_{\xv} f(\xv) + \frac{\rho}{2} \|\xv - (\vv^{(k)}-\frac{1}{\rho} \uv^{(k)})\|_2^2,  \label{Eq:solvex}\\
\vv^{(k+1)} &= \argmin_{\vv} \lambda g(\vv) + \frac{\rho}{2}\|\vv - (\xv^{(k)}+\frac{1}{\rho} \uv^{(k)})\|_2^2, \label{Eq:solvev}\\
\uv^{(k+1)} &= \uv^{(k)} + \rho (\xv^{(k+1)} - \vv^{(k+1)}), \label{Eq:u_k+1}
\end{align}
where the superscript $^{(k)}$ denotes the iteration number.

While in SCI and other inversion problems, $f(\xv)$ is usually a quadratic form and there are various solutions to Eq.~\eqref{Eq:solvex}. In PnP-ADMM, the solution of Eq.~\eqref{Eq:solvev} is replaced by an off-the-shelf denoising algorithm, to yield
\begin{equation}
{\vv^{(k+1)} = {\cal D}_{\sigma} (\xv^{(k)}+\frac{1}{\rho} \uv^{(k)})}.
\end{equation}
where ${\cal D}_{\sigma}$ denotes the denoiser being used with $\sigma$ being the standard deviation of the assumed additive white  Gaussian noise.
In~\cite{Chan2017PlugandPlayAF}, the authors proposed to update the $\rho$ in each iteration by $\rho_{k+1} = \gamma_k \rho_k$ with $\gamma_k \ge 1$ and setting $\sigma_k = \sqrt{\lambda/\rho_k}$ for the denoiser. In this manner, the author defined the {\em bounded denoiser} and proved the {\em fixed point} convergence of the PnP-ADMM.
\begin{definition} (Bounded Denoiser~\cite{Chan2017PlugandPlayAF}): A bounded denoiser with a parameter $\sigma$ is a function ${\cal D}_{\sigma}: {\mathbb R}^n \rightarrow {\mathbb R}^n$ such that for any input $\xv\in {\mathbb R}^{n}$, 
	\begin{equation}
	\frac{1}{n}\|{\cal D}_{\sigma}(\xv) - \xv\|_2^2 \le \sigma^2 C,
	\end{equation}
	for some universal constant $C$ independent of $n$ and $\sigma$. 
	\label{Definition1}
\end{definition}
\vspace{-1mm}
With this definition (constraint on the denoiser) and the assumption of $f:[0,1]^n \rightarrow {\mathbb R}$ having bounded gradient, which is for any $\xv \in [0,1]^n$, there exists $L < \infty$ such that $\|\nabla f(\xv)\|_2/\sqrt{n} \le L$, the authors of~\cite{Chan2017PlugandPlayAF} have proved that: 
the iterates of the PnP-ADMM demonstrates a fixed-point convergence. That is, 
there exists $(\xv^*, \vv^*, \uv^*)$ such that $\|\xv^{(k)} - \xv^*\|_2 \rightarrow 0$, $\|\vv^{(k)} - \vv^*\|_2 \rightarrow 0$, and $\|\uv^{(k)} - \uv^*\|_2 \rightarrow 0$ as $ k\rightarrow \infty$.
%

\subsection{PnP-ADMM for SCI}
In SCI, with the model stated in Eq.~\eqref{Eq:ghf}, $\xv \in {\mathbb R}^{nB}$, and we consider the loss function $f(\xv)$ as
\begin{equation}
\textstyle 	f(\xv) = \frac{1}{2}\|\yv - \Hmat \xv\|_2^2.
\end{equation}
Consider all the pixel values are normalized into $[0,1]$. 

\begin{lemma} In SCI, the function $f(\xv) = \frac{1}{2}\|\yv-\Hmat\xv\|_2^2$ has bounded gradients, \ie $\|\nabla f(\xv)\|_2\leq B \|\xv\|_2$. 
	\label{Lemma:fx_grad}
\end{lemma}
\begin{proof}
	The gradient of $f(\xv)$ in SCI is 
\begin{equation}
\nabla f(\xv) = \Hmat\ts\Hmat\xv-\Hmat\ts\yv,
\end{equation} where $\Hmat$ is a block diagonal matrix of size $n\times nB$ . 
\begin{list}{\labelitemi}{\leftmargin=12pt \topsep=0pt \parsep=0pt}
	\item The $\Hmat\ts\yv$ is a non-negative constant since both the measurement $\yv$ and the mask are non-negative in nature.
	\item Now let's focus on $\Hmat\ts\Hmat\xv$. Since
	\begin{align}
	\label{eq_sesci_PTP}
	\Hmat\ts\Hmat&=
	\left[
	\begin{matrix}
	\Dmat_1 \\ 
	\vdots \\ 
	\Dmat_B
	\end{matrix}
	\right] \left[
	\begin{matrix}
	\Dmat_1 \dots \Dmat_B
	\end{matrix}
	\right]\\
	& = \left[
	\begin{matrix}
	\Dmat_1^2& \Dmat_1\Dmat_2 & \cdots & \Dmat_1\Dmat_B\\
	\Dmat_1 \Dmat_2& \Dmat^2_2 & \cdots & \Dmat_2\Dmat_B\\
	\vdots & \vdots & \ddots & \vdots\\
	\Dmat_1 \Dmat_B& \Dmat_2\Dmat_B & \cdots & \Dmat^2_B
	\end{matrix}
	\right],
	\end{align}
\end{list}
due to this special structure, $\Hmat\ts\Hmat\xv$ is the weighted sum of the $\xv$ and $\|\Hmat\ts\Hmat\xv\|_2\leq B C_{\rm max}\|\xv\|_2$, where $C_{\rm max}$ is the maximum value in the sensing matrix. Usually, the sensing matrix is normalized to $[0,1]$ and this leads to $C_{\rm max}=1$ and therefore $\|\Hmat\ts\Hmat\xv\|_2\leq B \|\xv\|_2$.

Thus, $\nabla f(\xv)$ is bounded.
Furthermore, 
\begin{itemize}
	\item If the mask element $D_{i,j}$ is drawn from a binary distribution with entries \{0,1\} with a property of $p_1 \in (0,1)$ being 1, then
	\begin{eqnarray}
	\|\Hmat\ts\Hmat\xv\|_2\leq p_1 B \|\xv\|_2
	\end{eqnarray}
	with a high probability; usually, $p_1 = 0.5$ and thus $\|\Hmat\ts\Hmat\xv\|_2\leq 0.5 B \|\xv\|_2$.
	\item If the mask element $D_{i,j}$ is drawn from a Gaussian distribution ${\cal N}(0, \sigma^2)$ as in~\cite{Jalali18ISIT,Jalali19TIT_SCI}, though it is not practical to get negative modulation (values of $D_{i,j}$) in hardware, 
	\begin{eqnarray}
	\|\Hmat\ts\Hmat\xv\|_2\leq  B\sigma^2 \|\xv\|_2\stackrel{\sigma = 1}{=} B\|\xv\|_2,
	\end{eqnarray}
	with a high probability.
\end{itemize}
\end{proof}

\vspace{-2mm}
Recall that in~\eqref{Eq:Hmat_strucutre}, $\{\Dmat_i\}_{i=1}^B$ is a diagonal matrix and we denote its diagonal elements by
\begin{equation}
\Dmat_i = {\rm diag} (D_{i,1}, \dots, D_{i,n}).
\end{equation}
Thereby, in SCI, $\Hmat\Hmat\ts$ is diagonal matrix, \ie
\begin{equation}
{\Rmat = \Hmat\Hmat\ts = {\rm diag}(R_1, \dots, R_n)},\label{eq:R}
\end{equation}
where $ R_{j} = \sum_{b=1}^{B} D^2_{b,j}, \forall j = 1,\dots,n$.
We define
\begin{align}
R_{\rm max} &\stackrel{\rm def}{=} \max(R_1, \dots, R_n)= \lambda_{\rm max}(\Hmat\Hmat\ts),\\
R_{\rm min} &\stackrel{\rm def}{=} \min(R_1, \dots, R_n) =\lambda_{\rm min}(\Hmat\Hmat\ts),
\end{align}
where $\lambda_{\rm min}(\cdot)$ and $\lambda_{\rm max}(\cdot)$ represent the minimum and maximum eigenvalues of the ensured matrix.
%
\begin{assumption} \label{Ass:1}
	We assume that $\{R_j\}_{j=1}^n >0$. This means for each spatial location $j$, the $B$-frame modulation masks at this location have at least one non-zero entries. We further assume $R_{\rm max} > R_{\rm min}$. 
\end{assumption}
\vspace{-2mm}
This assumption makes senses in hardware as we expect at least one out of the $B$ frames is captured for each pixel during the sensing process.
Lemma~\ref{Lemma:fx_grad} along with the bounded denoiser in Definition~\ref{Definition1} give us the following Corollary.
\begin{corollary}
\label{Coro1}
	Consider the sensing model of SCI in \eqref{Eq:ghf} and further assume the elements in the sensing matrix satisfying Assumption~\ref{Ass:1}. Given $\{\Hmat,\yv\}$, ${\xv}$ is solved iteratively via PnP-ADMM with bounded denoiser, then $\xv^{(k)}$ and $\thetav^{(k)}$ will converge to a fixed point.
\end{corollary}
\begin{proof}
The proof follows \cite{Chan2017PlugandPlayAF} and thus omitted here.
\end{proof}

\section{Plug-and-Play GAP for SCI \label{Sec:PnP_GAP}}
\vspace{-1mm}
In this section, following the GAP algorithm~\cite{Liao14GAP} and the above conditions on PnP-ADMM, we propose PnP-GAP for SCI, which as mentioned before, has a lower computational workload (and thus faster) than PnP-ADMM. 
\begin{algorithm}[!htbp]
	\caption{Plug-and-Play GAP}
	\begin{algorithmic}[1]
		\REQUIRE$\Hmat$, $\yv$.
		\STATE Initial $\vv^{(0)}$, $\lambda_0$, $\xi<1$.
		\WHILE{Not Converge}
		\STATE Update $\xv$ by Eq.~\eqref{Eq:x_k+1}. 
		\STATE Update $\vv$ by denoiser  $\vv^{(k+1)} = {\cal D}_{\sigma_k}(\xv^{(k+1)})$.
		\IF {$\Delta_{k+1}\ge \eta \Delta_k$}
		\STATE {$\lambda_{k+1} = \xi \lambda_k$,}
		\ELSE 
		\STATE {$\lambda_{k+1} =  \lambda_k$.}
		\ENDIF
		\ENDWHILE
	\end{algorithmic}
	\label{algo:PP_GAP}
\end{algorithm}
\subsection{Algorithm}
Different from the ADMM in Eq.~\eqref{Eq:ADMM_xv}, GAP solves SCI by the following problem
\begin{equation} \label{Eq:GAP_xv}
({\hat \xv}, {\hat \vv}) = \argmin_{\xv,\vv} \frac{1}{2}\|\xv - \vv\|_2^2 + \lambda g(\vv), ~{\text{s.t.}}~~ \yv = \Hmat\xv.
\end{equation}
Similarly to ADMM, the minimizer in Eq.~\eqref{Eq:GAP_xv} is solved by a sequence of (now 2) subproblems and we again let $k$ denotes the iteration number.
\begin{itemize}
	\item Solving $\xv$: given $\vv$, $\xv^{(k+1))}$ is updated via an Euclidean projection of
	$\vv^{(k)}$ on the linear manifold ${\cal M}: \yv = \Hmat \xv$,
	\begin{equation}
	\xv^{(k+1)} =  \vv^{(k)} + \Hmat\ts (\Hmat \Hmat\ts)\inv (\yv - \Hmat \vv^{(k)}), \label{Eq:x_k+1}
	\end{equation}
	where as defined in~\eqref{eq:R}, $(\Hmat \Hmat\ts)\inv$ is fortunately a diagonal matrix and this has been observed and used in a number of SCI inversion problems. 
	\item Solving $\vv$: given $\xv$, updating $\vv$ can be seen as an denoising problem and
	\begin{equation}
	{\textstyle \vv^{(k+1)} = {\cal D}_{\sigma}(\xv^{(k+1)}).} \label{Eq:Denoise_GAP}
	\end{equation}
	Here, various denoiser can be used with $\sigma = \sqrt{\lambda}$.
\end{itemize}  
We can see that in each iteration, the only parameter to be tuned is $\lambda$ and we thus set $\lambda_{k+1} = \xi_k \lambda_k$ with $\xi_k\le 1$.
Inspired by the PnP-ADMM, we update $\lambda$ by the following two rules:
\begin{list}{\labelitemi}{\leftmargin=12pt \topsep=0pt \parsep=0pt}
	\item [a)] Monotone update by setting
	$\lambda_{k+1} = \xi \lambda_k$, with $\xi<1$. 
	\item [b)] Adaptive update
  by considering the relative residue:
	\begin{eqnarray}
	{\textstyle \Delta_{k+1} = \frac{1}{\sqrt{nB}}\left(\|\xv^{(k+1)} - \xv^{(k)}\|_2 + \|\vv^{(k+1)} - \vv^{(k)}\|_2\right)}.\nonumber \label{eq:Delta}
	\end{eqnarray}
	For any $\eta \in [0,1)$ and let $\xi<1$ be a constant, $\lambda_k$ is conditionally updated according  to the following settings:
	\begin{list}{\labelitemi}{\leftmargin=14pt \topsep=0pt \parsep=0pt}
		\item [i)] If $\Delta_{k+1}\ge \eta \Delta_k$, then $\lambda_{k+1} = \xi \lambda_k$.
		\item [ii)] If $\Delta_{k+1}< \eta \Delta_k$, then $\lambda_{k+1} =  \lambda_k$.
	\end{list}
\end{list}
With this adaptive updating of $\lambda_k$, the full PnP-GAP algorithm for SCI is exhibited in Algorithm~\ref{algo:PP_GAP}.

\subsection{Convergence}
\begin{assumption}\label{Ass:non_in}
	(Non-increasing denoiser) The denoiser in each iteration of PnP-GAP ${\cal D}_{\sigma_{k}}: {\mathbb R}^{nB} \rightarrow {\mathbb R}^{nB}$ performs denoising in a non-increasing order, \ie, $\sigma_{k+1}\le \sigma_k$ ($\lambda_{k+1}\le \lambda_k$). Further, when $k\rightarrow+\infty$, $\sigma_k \rightarrow 0$.
\end{assumption} 
This assumption makes sense since as the algorithm proceeds we expect the algorithm's estimate of the underlying signal to become more accurate, which means that the denoiser needs to deal with a less noisy signal.  
This is also guaranteed by the $\lambda$ setting in Algorithm~\ref{algo:PP_GAP}. 
With this assumption, we have the following convergence result of GAP-net.
\begin{theorem} \label{The:GAP_SCI_bound}
Consider the sensing model of SCI. Given $\{\Hmat,\yv\}$, ${\xv}$ is solved by PnP-GAP with bounded denoiser in a non-increasing order, then $\xv^{(k)}$ converges.
\end{theorem}
\begin{proof}
	From \eqref{Eq:x_k+1}, $\xv^{(k+1)} =  \vv^{(k)} + \Hmat\ts (\Hmat \Hmat\ts)\inv (\yv - \Hmat \vv^{(k)})$, we have
	\begin{equation}
	\xv^{(k+1)}-\xv^{(k)} =  \vv^{(k)}-\xv^{(k)} + \Hmat\ts (\Hmat \Hmat\ts)\inv (\yv - \Hmat \vv^{(k)}). 
	\end{equation}
	Following this, 
	\begin{align}
	&\|\xv^{(k+1)} - \xv^{(k)}\|_2^2 \nonumber\\
	=&\|\vv^{(k)} + \Hmat\ts (\Hmat \Hmat\ts)\inv (\yv - \Hmat \vv^{(k)}) - \xv^{(k)} \|^2_2 \\
	=& \|\vv^{(k)} + \Hmat\ts (\Hmat \Hmat\ts)\inv (\Hmat\xv^{(k)} - \Hmat \vv^{(k)}) - \xv^{(k)} \|^2_2 \nonumber\\
	=& \|(\Imat - \Hmat\ts (\Hmat \Hmat\ts)\inv \Hmat) (\vv^{(k)} - \xv^{(k)})\|^2_2  \nonumber\\
	=& \|\vv^{(k)} - \xv^{(k)}\|_2^2 - \|\Rmat^{-\frac{1}{2}}\Hmat (\vv^{(k)} - \xv^{(k)})\|_2^2 \label{Eq:xk_vkminus1}\\
	\le & \|\vv^{(k)} - \xv^{(k)}\|_2^2 \\
	=&  \| {\cal D}_{\sigma_{k}} (\xv^{(k)}) - \xv^{(k)}\|_2^2  \\
	\le&  \sigma_k^2 nBC \label{Eq:convg_C},
	\end{align} 
	where in \eqref{Eq:xk_vkminus1} $\Rmat = \Hmat\Hmat\ts$ and in \eqref{Eq:convg_C} we have used the bounded denoiser.
	Using Assumption \ref{Ass:non_in} (non-increasing denoiser), we have
	$\sigma_k \rightarrow 0$,  $\|\xv^{(k+1)} - \xv^{(k)}\|_2^2 \rightarrow 0$ and thus $\xv^{(k)}$ converges.
\end{proof}

\subsection{PnP-ADMM vs. PnP-GAP}
Comparing PnP-GAP in Eqs~\eqref{Eq:x_k+1}-\eqref{Eq:Denoise_GAP} and PnP-ADMM in Eqs~\eqref{Eq:solvex}-\eqref{Eq:u_k+1}, we can see that PnP-GAP only has two subproblems (rather than three as in PnP-ADMM) and thus the computation is faster. 
It was pointed out in~\cite{Liu18TPAMI} that in the noise-free case, ADMM and GAP perform the same and this has been mathematically proved. However, in the noisy case, ADMM usually performs better since it considered noise in the model and below we give a geometrical explanation.

As shown in Fig.~\ref{fig:ADMM_GAP}, where we used a two-dimensional sparse signal as an example, we can see that since GAP imposes $\yv=\Hmat \hat\xv$, the solution of GAP $\hat\xv$ is always on the dash-green line (due to noise, this line might be deviated from the solid green line where the true single lies on). However, the solution of ADMM does not have this constraint but to minimize $\|\yv-\Hmat\xv\|_2^2$, which can be in the dash-red circle or the yellow-dash circle depending on the initialization. In the noise-free case, both GAP and ADMM will have a large chance to converge to the true signal $\xv^*$. However, in the noisy case, the Euclidean distance between GAP solution and the true signal ($\|\hat{\xv} - \xv^*\|_2$) might be larger than that of ADMM. Again, the final solution of ADMM depends on the  initialization and it is not guaranteed to be more accurate than that of GAP.

\begin{figure}
	\begin{center}
		\includegraphics[width=\linewidth]{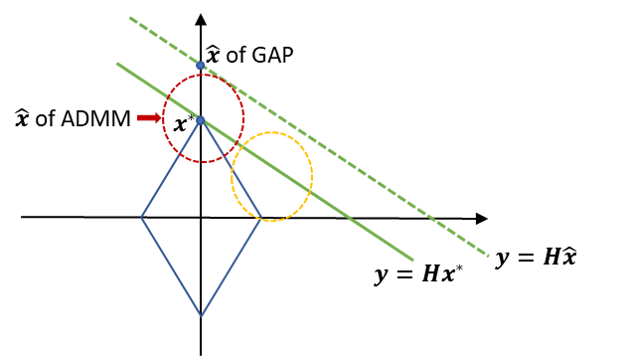}
	\end{center}
	\caption{Demonstration of the solution of ADMM (within the dash-red circle or dash-yellow circle depending on the initialization) and GAP ($\hat\xv$) under the noisy case. The difference is that the solution of GAP always lies on $\yv = \Hmat \hat\xv$.}
	\label{fig:ADMM_GAP}
\end{figure}

\begin{table*}[!htbp]
	\caption{The average results of PSNR in dB (left entry in each cell) and SSIM (right entry in each cell) and run time per measurement/shot in minutes by different algorithms on 6 benchmark datasets.}
	\centering
	\resizebox{2.07\columnwidth}{!}
	{
		\begin{tabular}{cV{3}ccccccV{3}cc}
			\hlineB{3}
			Algorithm& \texttt{Kobe} & \texttt{Traffic} & \texttt{Runner} & \texttt{Drop} & \texttt{Crash} & \texttt{Aerial} & Average &  Run time (min) \\
			\hlineB{3}
			GAP-TV           & 26.46, 0.8848 & 20.89, 0.7148 & 28.52, 0.9092 & 34.63, 0.9704 & 24.82, 0.8383 & 25.05, 0.8281 & 26.73, 0.8576 & 0.07 \\
			\hline\textbf{}
			{DeSCI (GAP-WNNM)} & {\bf 33.25}, {0.9518} & {\bf 28.71}, {\bf 0.9250} & {\bf 38.48}, {\bf 0.9693} & 43.10, 0.9925 & {\bf 27.04}, {\bf 0.9094} & 25.33, 0.8603 & {\bf 32.65}, {\bf 0.9347} & 103.0 \\
			\hlineB{3}
			PnP-VBM4D        & 30.60, 0.9260 & 26.60, 0.8958 & 30.10, 0.9271 & 26.58, 0.8777 & 25.30, 0.8502 & 26.89, 0.8521 & 27.68, 0.8882 & 7.9  \\
			\hline
			\rowcolor{lightgray}
			PnP-FFDNet       & 30.50, 0.9256 & 24.18, 0.8279 & 32.15, 0.9332 & 40.70, 0.9892 & 25.42, 0.8493 & 25.27, 0.8291 & 29.70, 0.8924 & {{\bf 0.05} (GPU)} \\
			\hline
			
			PnP-WNNM-TV      & 33.00, 0.9520 & 26.76, 0.9035 & 38.00, 0.9690 & 43.27, 0.9927 & 26.25, 0.8972 & 25.53, 0.8595 & 32.14, 0.9290 & 40.8 \\
			\hline
			PnP-WNNM-VBM4D   & 33.08, \bf{0.9537} & 28.05, 0.9191 & 33.73, 0.9632 & 28.82, 0.9289 & 26.56, 0.8874 & {\bf 27.74}, {\bf 0.8852} & 29.66, 0.9229 & 25.0 \\
			\hline
			PnP-WNNM-FFDNet & 32.54, 0.9511 & 26.00, 0.8861 & 36.31, 0.9664 & \bf{43.45}, \bf{0.9930} & 26.21, 0.8930 & 25.83, 0.8618 & 31.72, 0.9252 & 17.9 \\
			\hlineB{3}
	\end{tabular}}
	\label{Tab:results_4video}
\end{table*}

\section{Integrate Various Denoisers into PnP for SCI Reconstruction\label{Sec:P3}}
It can be seen from Theorem~\ref{The:GAP_SCI_bound} that the reconstruction error term depends on the bounded denoising algorithm. 
In other words, a better denoiser with a smaller $C$ can provide a reconstruction result closer to the true signal.
Various denoising algorithms exist for different tasks based on speed and quality.
Usually, a faster denoiser~\eg, TV, is very efficient, but cannot provide high-quality results.  
The middle class algorithms~\eg, BM3D~\cite{Dabov07BM3D} can provide decent results with a longer running time.
More advanced denoising algorithm,~\eg, WNNM~\cite{Gu14CVPR,Gu17IJCV} can provide better results~\cite{Liu18TPAMI}, but even slower.
Another line of emerging denoising approaches is based on deep learning~\cite{XieNIPS2012_deepDN,zhang2017beyond}, which can provide decent results within a short time after training, but they are usually not robust to noise levels and in high noisy cases, the results are not good.
Different from conventional denoising problems, in SCI reconstruction, the noise level in each iteration is usually from large to small and the dynamic range can from 150 to 1, considering the pixel values within $\{0,1,\dots, 255\}$. 
Fortunately, FFDNet~\cite{Zhang18TIP_FFDNet} has provided us a fast and flexible solution under various noise levels. 

By integrating these denoising algorithms into PnP-GAP/ADMM, we can have different algorithms (Table~\ref{Tab:results_4video} and Fig.~\ref{fig:demo}) with  different results. It is worth noting that DeSCI can be seen as PnP-WNNM, and its best results are achieved by exploiting the correlation across different video frames. 
On the other hand, most existing deep denoising priors are still based on images. Therefore, it is expected that the results of PnP-GAP/ADMM-FFDNet are not as good as DeSCI. We anticipate that with the advances of deep denoising priors, better video denoising method will boost the our PnP-based SCI reconstruction results.
In addition, these different denoisers can be used in parallel, \ie, one after each other in one GAP/ADMM iteration or used sequentially, \ie, the first $K_1$ iterations using FFDNet and the next $K_2$ iterations using WNNM to achieve better results.

It is worth noting that by assuming WNNM being a bounded denoiser, DeSCI~\cite{Liu18TPAMI}, which is GAP-WNNM, is a special case of our PnP-GAP.


\section{Results \label{Sec:results}}
We applied the proposed PnP algorithms to both simulation~\cite{Liu18TPAMI,Ma19ICCV} and real datasets captured by the SCI cameras~\cite{Patrick13OE,Yuan14CVPR}.
Conventional denoising algorithms  include TV~\cite{Yuan16ICIP_GAP}, VBM4D~\cite{Maggioni2012VideoDD} and WNNM~\cite{Gu14CVPR} are used. For the deep learning based denoiser, we have tried various algorithms and found that FFDNet~\cite{Zhang18TIP_FFDNet} provides the best results.

\paragraph{Simulation: Benchmark Data}
We follow the simulation setup in~\cite{Liu18TPAMI} using the six datasets, \ie, \texttt{Kobe, Traffic, Runner, Drop, crash,} and \texttt{aerial}~\cite{Ma19ICCV}\footnote{The results of DeSCI (GAP-WNNM) is different from those reported in \cite{Ma19ICCV} because of parameter settings of DeSCI, specifically the input estimated noise levels for each iteration stage. We use exactly the same parameters as the DeSCI paper~\cite{Liu18TPAMI}, which is publicly available at \href{https://github.com/liuyang12/DeSCI}{https://github.com/liuyang12/DeSCI}.}, where $B=8$ video frames are compressed into a single measurement.
Table~\ref{Tab:results_4video} summarizes the PSNR and SSIM~\cite{Wang04imagequality} results of these 6 benchmark data using various denoising algorithms, where DeSCI can be categorized as GAP-WNNM, and PnP-WNNM-FFDNet used 50 iterations FFDNet and then 60 iterations WNNM, similar for GAP-WNNM-VBM4D.
It can be observed that:
\begin{itemize}
	\item [$i$)] By using GPU, PnP-FFDNet is now the fastest algorithm\footnote{{Only a regular GPU is needed to run FFDNet and since FFDNet is performed in a frame-wise manner, we do not need a large amount of CPU or GPU RAM (no more than 2GB here) compared to other video denoisers using parallization (even with parallelization, other algorithms listed here are unlikely to outperform PnP-FFDNet in terms of speed).}}; it is even faster than GAP-TV, meanwhile providing about 3dB higher PSNR than GAP-TV. Therefore, PnP-FFDNet can be used as {\em an efficient baseline} in SCI reconstruction. Since the average PSNR is close to 30dB, it is applicable in real cases. This will be further verified in the following subsection on large-scale datasets.
	\item [$ii$)]
	DeSCI still provides the best results on average; however, by combing other algorithms with WNNM, comparable results (\eg PnP-WNNM-FFDNet) can be achieved by only using $1/6$ computational time.
\end{itemize}

Fig.~\ref{fig:comp_frames_full} plots selected frames of the six datasets using different algorithms. It can be seen that though GAP-WNNM still leads to best results, the difference between PnP-FFDNet and DeSCI is very small and in most cases, they are close to each other. 

It can be seen clearly that PnP-FFDNet provides overall comparable results as the state-of-the-art (best among all the seven methods listed here) method DeSCI, as shown in Fig.~\ref{fig:comp_frames_full} with significantly reduced running time (3 seconds vs. 103 minutes). 

\begin{figure*}[!htbp]
	\begin{center}
		\includegraphics[width=1.0\linewidth]{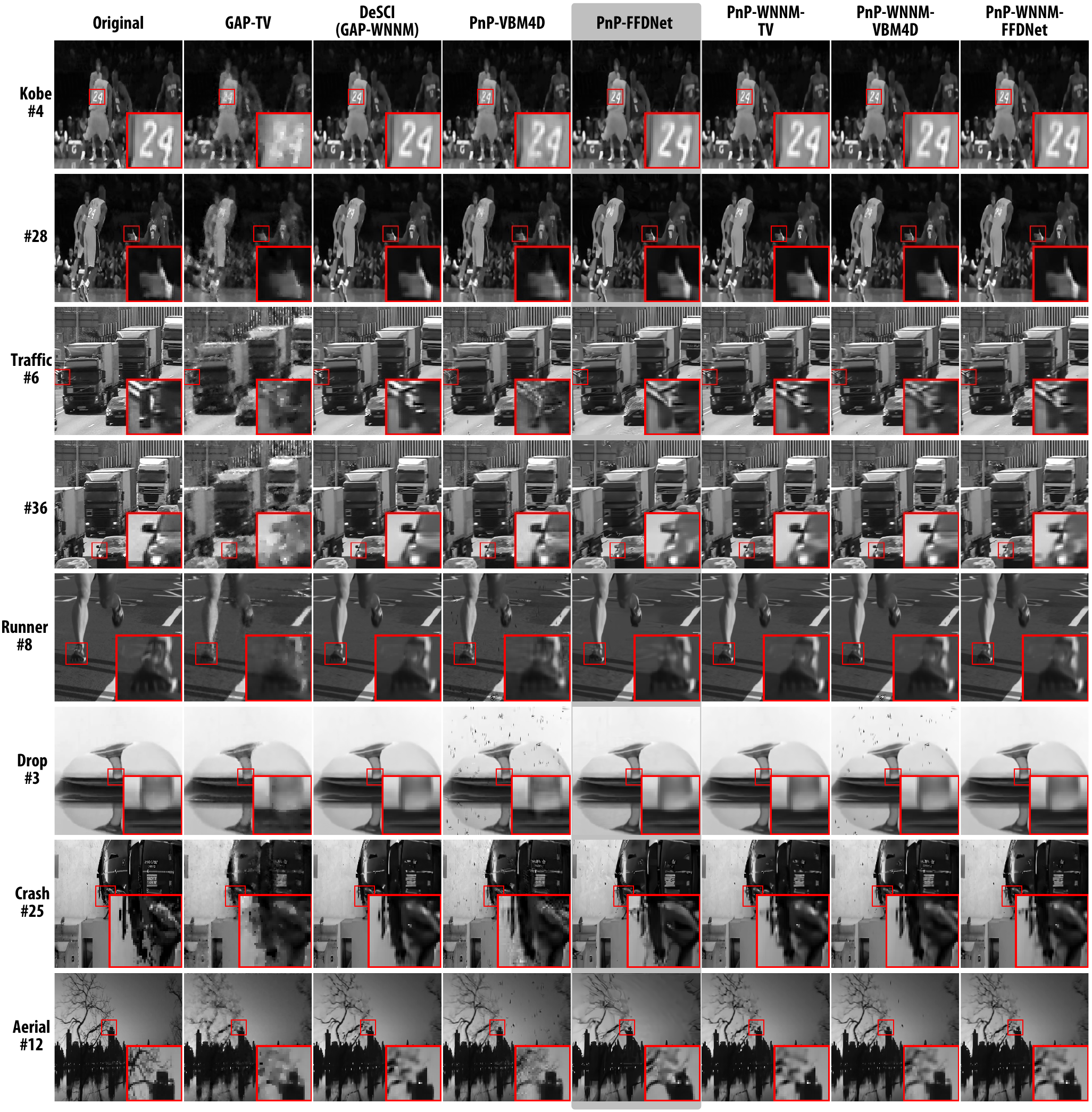}
	\end{center}
	\vspace{-3mm}
	\caption{Full comparison of reconstructed frames of PnP-GAP algorithms (GAP-TV, DeSCI (GAP-WNNM), PnP-VBM4D, PnP-FFDNet, PnP-WNNM-TV, PnP-WNNM-VBM4D, and PnP-WNNM-FFDNet) on six simulated video SCI datasets.}
	\label{fig:comp_frames_full}
\end{figure*}


\paragraph{Simulation: Large-scale Data}
Similar to the benchmark data, we simulate the color video SCI measurements for large-scale data with four YouTube slow-motion videos, \ie, \texttt{Messi}\footnote{\href{https://www.youtube.com/watch?v=sbPrevs6Pd4}{https://www.youtube.com/watch?v=sbPrevs6Pd4}}, \texttt{Hummingbird}\footnote{\href{https://www.youtube.com/watch?v=RtUQ_pz5wlo}{https://www.youtube.com/watch?v=RtUQ\_pz5wlo}}, \texttt{Swinger}\footnote{\href{https://www.youtube.com/watch?v=cfnbyX9G5Rk}{https://www.youtube.com/watch?v=cfnbyX9G5Rk}}, and \texttt{Football}\footnote{\href{https://www.youtube.com/watch?v=EGAuWZYe2No}{https://www.youtube.com/watch?v=EGAuWZYe2No}}. 

The color video SCI system and sensing process follows the color video and depth SCI system in \cite{Yuan14CVPR}. The scematic of a color video SCI system is shown in Fig.~\ref{fig:video_color_sci}. A sequence of color scene is coded by the corresponding shifted random binary masks at each time step and finally summed up to form a shapshot measurement on the color Bayer RGB sensor (with a ``RGGB'' Bayer color filter array here).

For reconstruction, the snapshot measurement is splitted into four ``RGGB'' sub-measurements according to the Bayer pattern. These sub-measurements are reconstructed measurement-by-measurement following the gray-scale reconstruction process by iteratively update the signal in data domain (using GAP or ADMM) and prior domain (using plug-and-play denoisers). Finally, the reconstructed sub-video-frames representing different color channels (R, G1, G2, and B) are recombined to a mosaic image and then demosaiced to form full-color video frames. Note that for simulation of large-scale data using YouTube videos, we do not have the access to the raw video data before demosaicing, so we simply ``up-sample'' it by putting each color channel as the mosaic R, G1, G2, and B channels. In this way, there are two identical G channels here and the reconstructed and the size of demosaiced image is doubled (both in width and height). For example, for UHD color video \texttt{Football} with original image size of $3840\times1644$, the reconstructed video frames have the size of $7680\times3288$ (demosaiced). To make the readers less confusing, we simply call it UHD here (8K UHD exactly). And the quantitative metrics (PSNR and SSIM) are calculated before demosaicing.

\begin{itemize}
	\item \href{./largescale\_messi24.avi}{largescale\_messi24.avi}: A $1920\times1080\times24$ video reconstructed from a snapshot.
	\item \href{./largescale\_hummingbird40.avi}{largescale\_hummingbird40.avi}: A $1920\times1080\times40$ video reconstructed from a snapshot.
	\item \href{./largescale\_swinger20.avi}{largescale\_swinger20.avi}: A $3840\times2160\times20$ video reconstructed from a snapshot.
	\item \href{./largescale\_football48.avi}{largescale\_football48.avi}: A $3840\times1644\times48$ video reconstructed from a snapshot.
\end{itemize}

\begin{figure*}[!htbp]
	\begin{center}
		\includegraphics[width=0.9\linewidth]{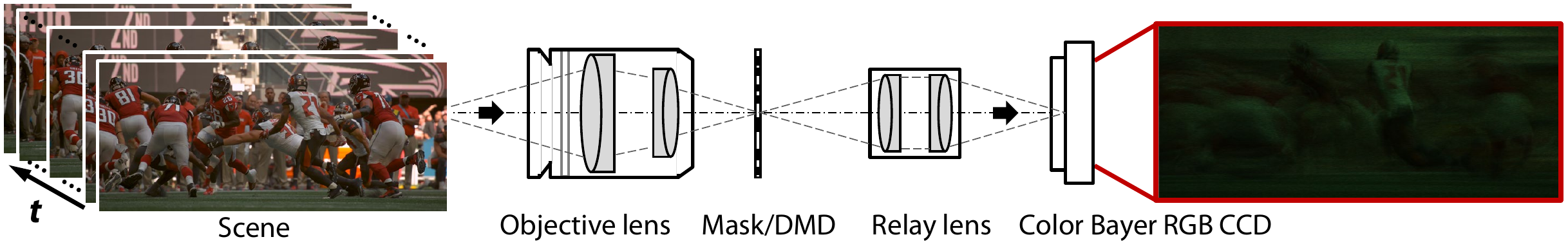}
	\end{center}
	\caption{Schematic of a color video SCI system and its snapshot measurement (showing in Bayer RGB mode). A ``RGGB'' Bayer pattern is shown here.}
	\label{fig:video_color_sci}
\end{figure*}

\paragraph{Performance varying compression ratios ($\Bmat$)} 
\label{par:performance_varying_compression_ratio_}
In order to further illustrate the efficiency of the proposed PnP algorithms for SCI, especially in real SCI systems varying compression ratios ($B$), we show the reconstruction quality and speed of three PnP-SCI algorithms with compression ratios from $8$ to $48$ in Fig.~\ref{fig:quality_vary_codenum} and Fig.~\ref{fig:speed_vary_codenum}, respectively. The data we used is the downsampled grayscale video of \texttt{Football} with pixel resoltion of 720p ($1280\times720$). The other algorithms listed in Tab.~1 are too slow to be compared. And other deep-learning-based end-to-end networks, like~\cite{Ma19ICCV,Qiao2020_APLP} would be not flexible to different compression ratios and require re-training the network for each compression ratio.

As we can see in Fig.~\ref{fig:quality_vary_codenum} and Fig.~\ref{fig:speed_vary_codenum}, PnP-FFDNet is of the highest quality and fastest speed among these three fast PnP-SCI algorithms even with high compression ratios (up to 48). This further supports the idea that PnP-FFDNet would be the baseline for SCI reconstruction. 

\begin{figure}[t]
	\begin{center}
		\includegraphics[width=0.9\linewidth]{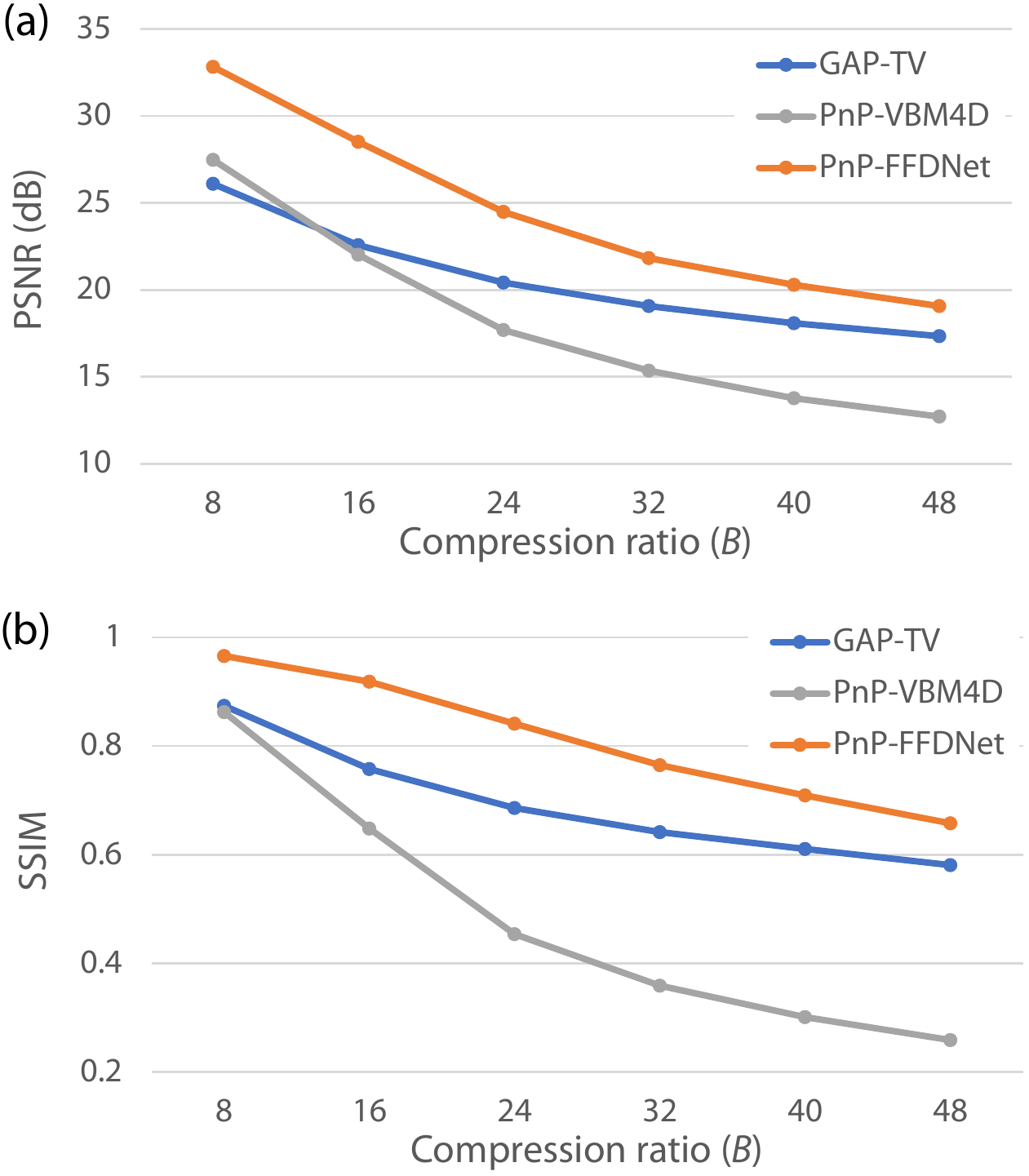}
	\end{center}
	\caption{Reconstruction quality, \ie, PSNR (a) and SSIM (b) varying compression ratios from 8 to 48. Higher is better.}
	\label{fig:quality_vary_codenum}
\end{figure}

\begin{figure}[t]
	\begin{center}
		\includegraphics[width=0.9\linewidth]{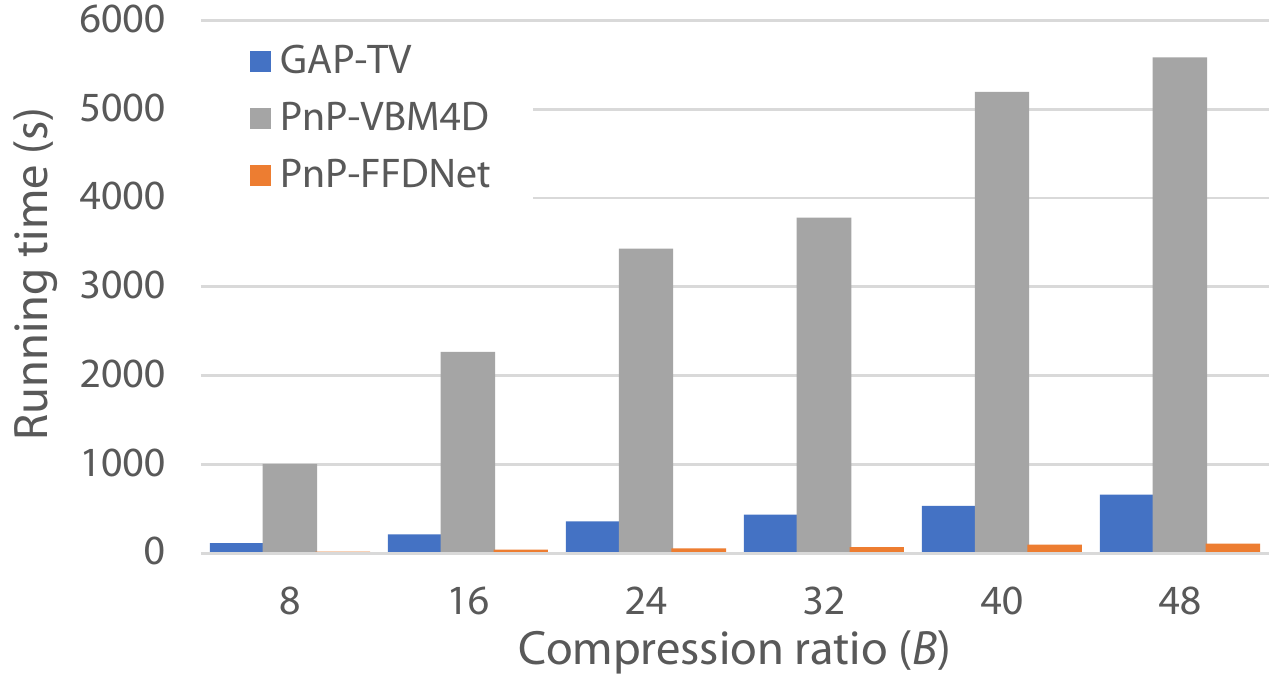}
	\end{center}
	\caption{Running time (in seconds) varying compression ratios from 8 to 48. Lower is better.}
	\label{fig:speed_vary_codenum}
\end{figure}

%
\begin{figure*}[!htbp]
	\centering
		\includegraphics[width=.97\linewidth]{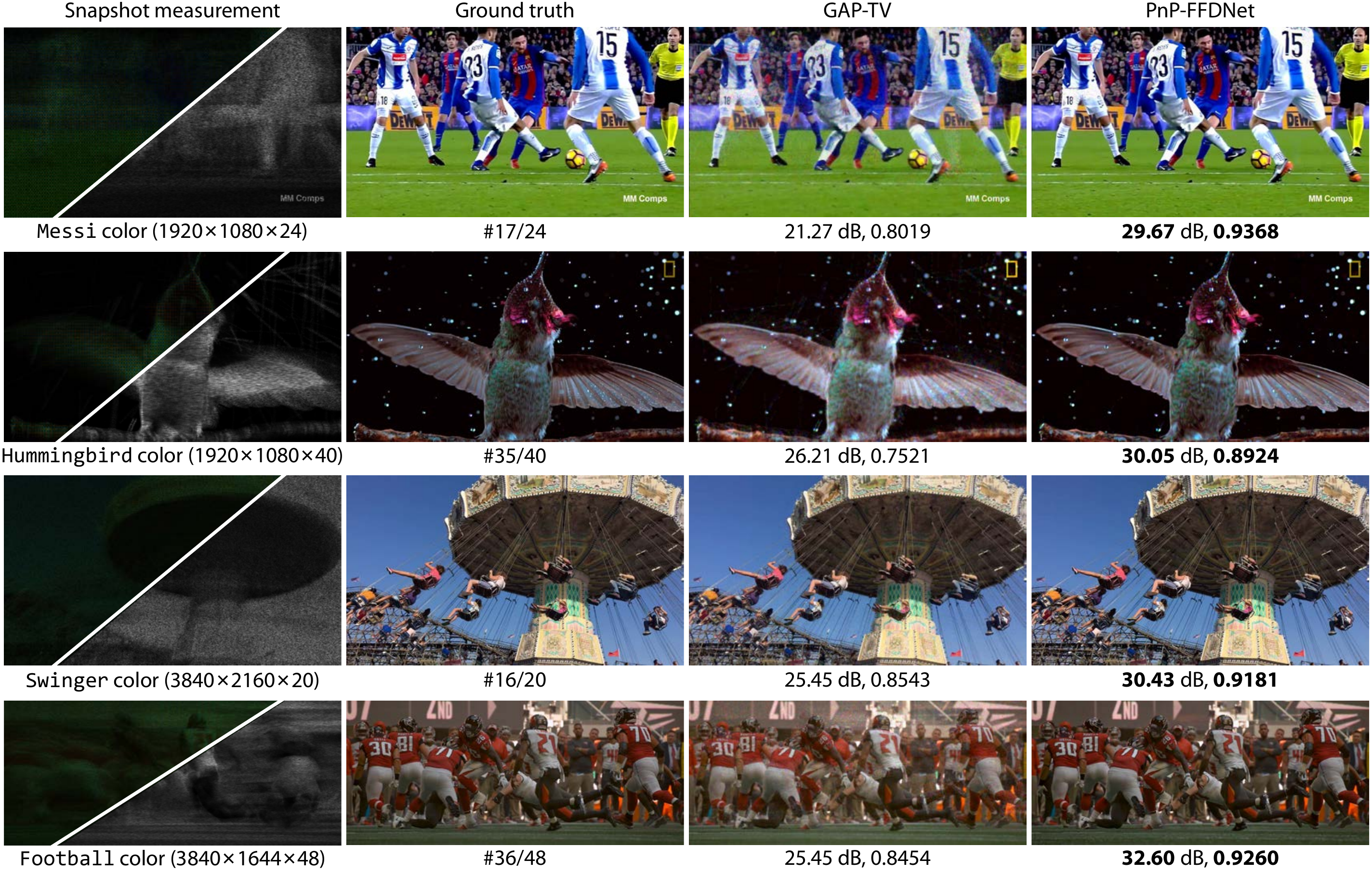}
	\caption{Reconstructed frames of PnP-GAP algorithms (GAP-TV and PnP-FFDNet) on four simulated large-scale video SCI datasets.}
	\label{fig:comp_largescale}
\end{figure*}
\begin{figure}[!htbp!]
	\begin{center}
		\includegraphics[width=.9\linewidth]{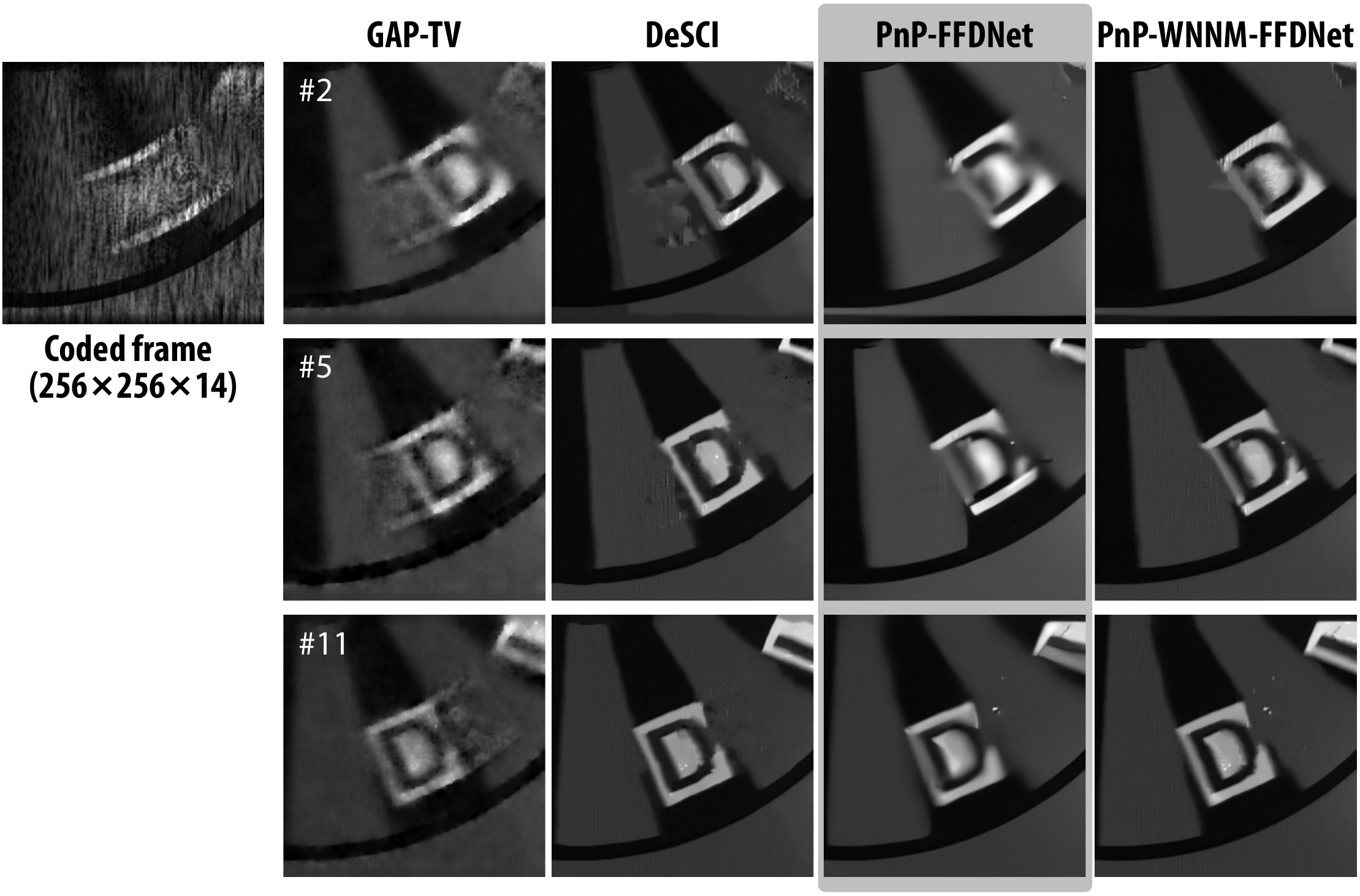}\\
	\end{center}
	\caption{Real data: \texttt{chopper wheel} ($256\times256\times14$).}
	\label{fig:real_chopperwheel}
\end{figure}
\begin{figure}[!htbp]
	\begin{center}
		\includegraphics[width=.9\linewidth]{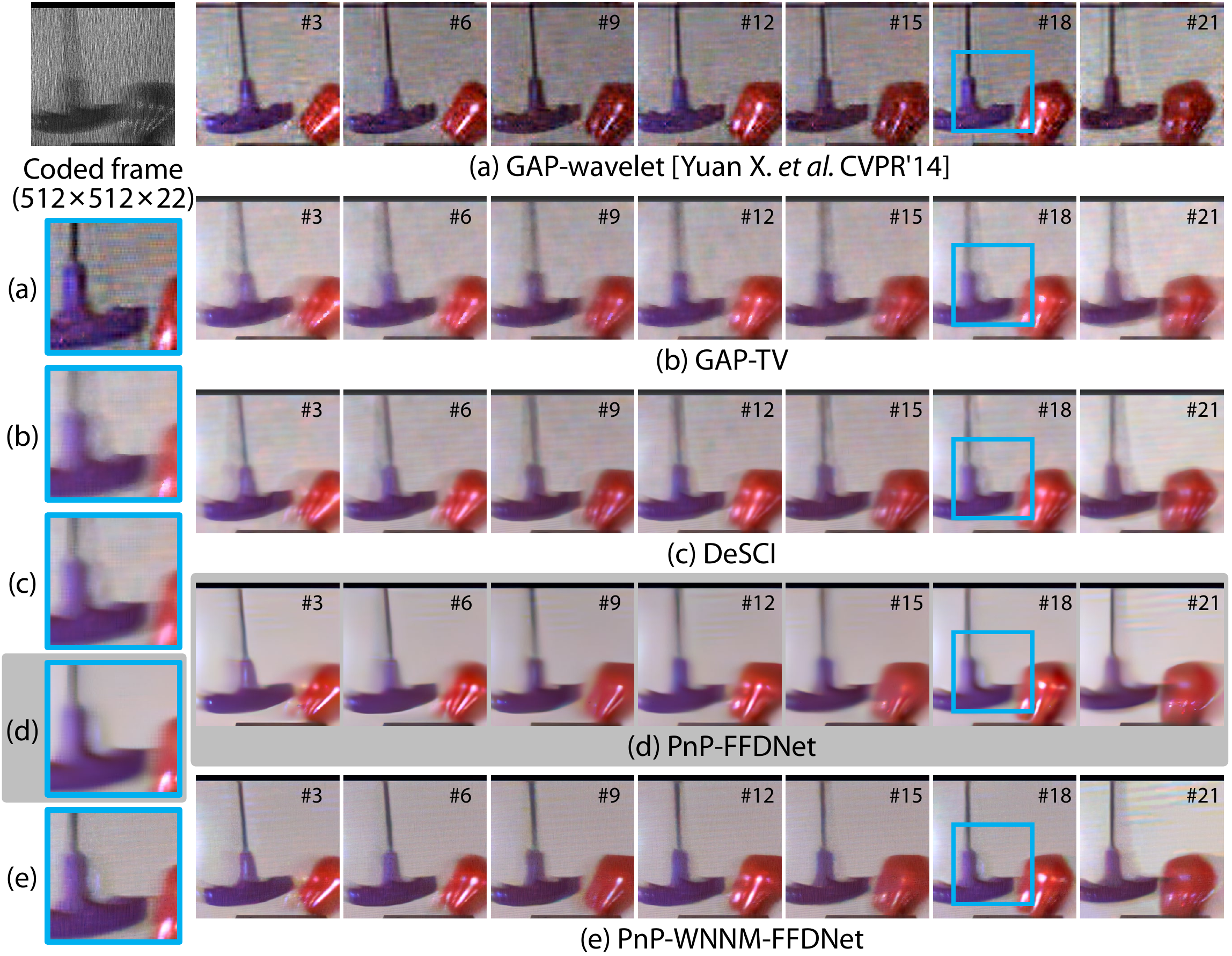}
	\end{center}
	\caption{Real data: \texttt{hammer} color video SCI ($512\times512\times22$).}
	\label{fig:real_color_hammer}
\end{figure}

\paragraph{Real Data}
Lastly, we apply the proposed PnP framework to real data captured by SCI cameras to verify the robustness of the algorithms. Figs.~\ref{fig:real_chopperwheel}-\ref{fig:real_color_hammer} show the results of different compression ratios and different sizes. It can be observed that in most cases, PnP-FFDNet can provide comparable or even better (\texttt{chopper wheel}) results than DeSCI but again with a significant saving on computational time. 
The running time of these data using different algorithms is shown in Table~\ref{Table:time_real}, where we can see that PnP-FFDNet provides results around 12 seconds even for a $512\times512\times22$ large-scale video.

We show more results of \texttt{labs}~\cite{Qiao2020_APLP} and \texttt{UCF}~\cite{Sun16OE} in Fig.~\ref{fig:real_labs} and Fig.~\ref{fig:real_ucf}, respectively. 
From Fig.~\ref{fig:real_labs} and Fig.~\ref{fig:real_ucf}, we can see that PnP-FFDNet, which only takes about 12 seconds for reconstruction, can provide comparable results as DeSCI, which needs hours even when performed in a frame-wise manner, as shown in Table~\ref{Table:time_real}. And PnP-FFDNet is significantly better than the speed runner-up GAP-TV in terms of motion-blur reduction and detail preservation, as shown in Figs.~\ref{fig:real_labs} and \ref{fig:real_ucf}. Note that PnP-FFDNet is more than $5\times$ faster than GAP-TV in real datasets with regular size, and more than $7\times$ faster in large-scale datasets. In this way, PnP algorithms for SCI achieves a good balance of efficiency and flexibility and PnP-FFDNet could serve as a baseline for SCI recovery. 

\begin{figure}[t]
	\begin{center}
		\includegraphics[width=1.0\linewidth]{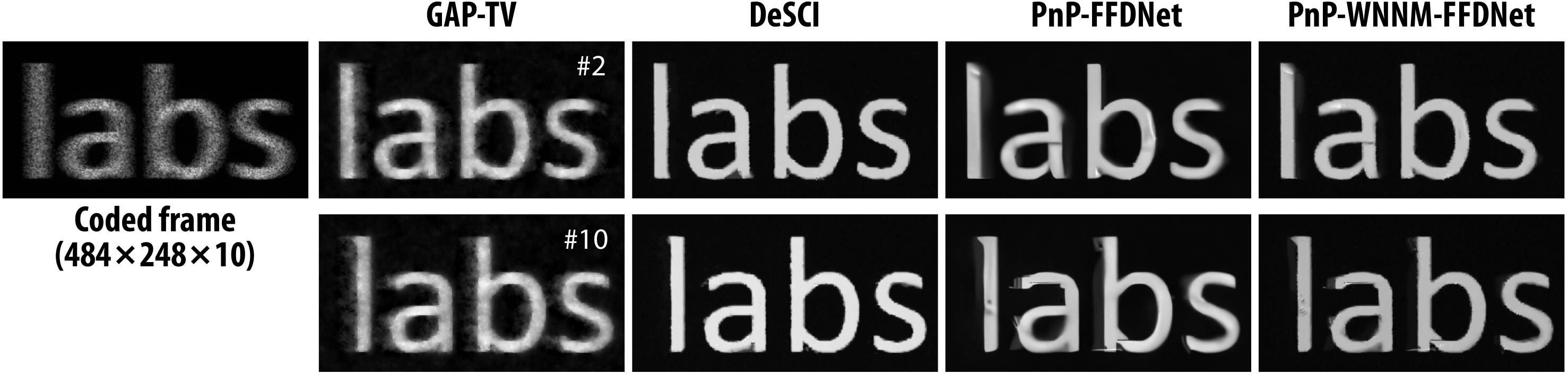}
	\end{center}
	\caption{Real data: \texttt{labs} high-speed video SCI ($484\times248\times10$).}
	\label{fig:real_labs}
\end{figure}

\begin{figure}[t]
	\begin{center}
		\includegraphics[width=1.0\linewidth]{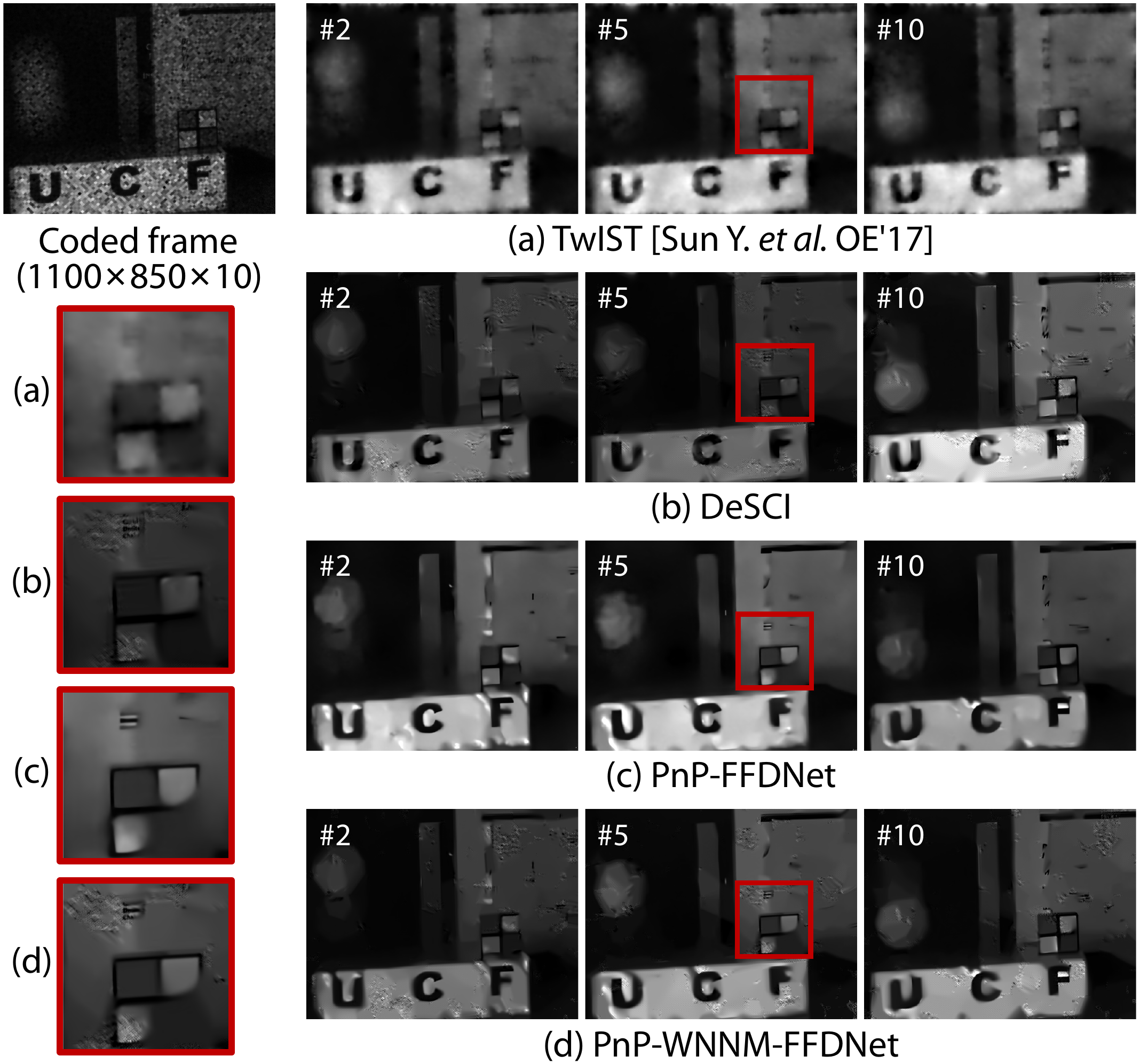}
	\end{center}
	\caption{Real data: \texttt{UCF} high-speed video SCI ($1100\times850\times10$).}
	\label{fig:real_ucf}
\end{figure}



\begin{table*}
\caption{{Running time (second) of real data using different algorithms. Visual results of \texttt{labs} and \texttt{UCF} are shown in the SM.}}
\resizebox{2.00\columnwidth}{!}
	{
	\begin{threeparttable}
\begin{tabular}{c cV{3}cc>{\columncolor[gray]{.8}[.5\tabcolsep]}cc}
\hlineB{3}
Real dataset  & Pixel resolution & {GAP-TV} & {DeSCI} & {PnP-FFDNet} & {PnP-WNNM-FFDNet} \\ \hlineB{3}
\texttt{chopperwheel} & $256\times256\times14$       & 11.6                        & 3185.8                     & \textbf{2.7}                             & 1754.7                               \\ \hline
\texttt{labs}         & $484\times248\times10$       & 36.9                        & 6471.3                     & \textbf{4.5}                             & 3226.5                               \\ \hline
\texttt{hammer} color & $512\times512\times22$       & 94.5                        & 4791.0                     & \textbf{12.6}                            & 1619.4                               \\ \hline
\texttt{UCF}          & $1100\times850\times10$      & 300.8                       & 2938.8*                    & \textbf{12.5}                            & 1504.5*            \\ \hlineB{3}
\end{tabular}
\begin{tablenotes}
  \item[*] WNNM is performed in a frame-wise manner for large-scale datasets.
  \end{tablenotes}
  \end{threeparttable}
}
\label{Table:time_real}
\end{table*}

\section{Conclusions \label{Sec:Con}}
We proposed plug-and-play algorithms for the reconstruction of snapshot compressive video imaging systems. By integrating deep denoisers into the PnP framework, we not only get excellent results on both simulation and real datasets, but also provide reconstruction in a short time with sufficient flexibility. 
Convergence results of PnP-GAP are proved and we first time show that SCI can be used in large-scale (HD, FHD and UHD) daily life videos. This paves the way of practical applications of SCI.

Regarding the future work, one direction is to train a better video (rather than image) denoising network and apply it to the proposed PnP framework to further improve the reconstruction results. The other direction is to build a real large-scale video SCI system to be used in advanced cameras~\cite{BradyNature12}.

\paragraph{Acknowledgments.} The work of Jinli Suo and Qionghai Dai is partially supported by NSFC 61722110, 61931012, 61631009 and Beijing Municipal Science \& Technology Commission (BMSTC) (No. Z181100003118014).

{\small

}

%
%

%
%

\end{document}